\documentclass[11pt]{article}
\usepackage{fullpage,amsmath,amsfonts,comment} 
\usepackage{amsthm,amssymb,framed}
\usepackage{times}
\usepackage[colorlinks=true,urlcolor=magenta,linkcolor=blue,citecolor=blue]{hyperref}

\usepackage{cleveref}

\usepackage{algorithm}
\usepackage{algpseudocode}
\usepackage{mathtools}
\usepackage{todonotes}
\usepackage{amsmath}
\usepackage{amsthm}
\usepackage{thmtools}
\usepackage{graphicx}
\usepackage{tikz}
\usepackage[T1]{fontenc}

\newtheorem{property}{Property}

\usepackage{dsfont}

\usepackage{amsmath}

\def\poly{\text{poly}}

\newcommand{\lam}{\lambda}
\newcommand{\NP}{\textsf{NP}}
\newcommand{\DC}{\mathcal{C}}
\newcommand{\A}{\mathcal{A}}
\newcommand{\OC}{O_{\DC}}

\newcommand{\tup}[1]{\langle#1\rangle}

\makeatletter
\newcommand{\oset}[3][0ex]{%
  \mathrel{\mathop{#3}\limits^{
    \vbox to#1{\kern-6\ex@
    \hbox{$\scriptstyle#2$}\vss}}}}
\makeatother

\usepackage{tcolorbox}
\newcommand{\defproblem}[3]{
  \begin{tcolorbox}[
    colback=white, % Background color
    colframe=black, % Border color
    width=\textwidth, % Box width
    boxrule=0.2mm, % Border thickness
    arc=5pt, % Corner radius
    left=2mm, % Left padding
    right=2mm, % Right padding
    top=2mm, % Top padding
    bottom=2mm % Bottom padding
  ]
    \begin{minipage}{0.95\textwidth}
    {#1}\\
    ~\\
    {\bf{Input:}}~#2  \\
    {\bf{Output:}}~#3
    \end{minipage}
  \end{tcolorbox}
}

\usepackage{etoolbox}
\makeatletter
\patchcmd{\@bibitem}{\doi}{\small\doi}{}{}
\makeatother

\newtheorem{theorem}{Theorem}[section]

\newtheorem{lemma}[theorem]{Lemma}

\newtheorem{definition}[theorem]{Definition}

\title{Mind the Gap.~Doubling Constant Parametrization of Weighted Problems:~TSP, Max-Cut, and More\footnote{To appear at STACS 2026.}}

\author{Mihail Stoian\\University of Technology Nuremberg}

\date{}

\begin{document}

\maketitle

\begin{abstract}
Despite much research, hard weighted problems still resist super-polynomial improvements over their textbook solution.
On the other hand, the unweighted versions of these problems have recently witnessed the sought-after speedups.
Currently, the only way to repurpose the algorithm of the unweighted version for the weighted version is to employ a polynomial embedding of the input weights. This, however, introduces a pseudo-polynomial factor into the running time, which becomes impractical for arbitrarily weighted instances.

In this paper, we introduce a new way to repurpose the algorithm of the unweighted problem. Specifically, we show that the time complexity of several well-known \textsf{NP}-hard problems operating over the $(\min, +)$ and $(\max, +)$ semirings, such as TSP, Weighted Max-Cut, and Edge-Weighted $k$-Clique, is proportional to that of their unweighted versions when the set of input weights has small doubling. We achieve this by a meta-algorithm that converts the input weights into polynomially bounded integers using the recent constructive Freiman's theorem by Randolph and Węgrzycki [ESA~2024] before applying the polynomial embedding.
\end{abstract}

\thispagestyle{empty}
\setcounter{page}{0}

\newpage

\section{Introduction}

The renaissance witnessed in exact algorithm design over the last two decades has one of its roots in the quest to answer the following fundamental question:
\begin{quote}
    \centering
    \itshape
    Are textbook solutions of \NP-hard problems the best we can hope for?
\end{quote}
Indeed, there are many problems that resist improvements over the standard algorithm.
Yet, we have arguments to be optimistic: The amount of work in this direction has led to impressive speedups for well-studied problems such as \textsc{Max-Cut}~\cite{maxcut-williams}, \textsc{Hamiltonicity}~\cite{hamiltonicity, below_all_subsets1}, \textsc{Sched}~\cite{below_all_subsets2}, \textsc{Independent Set}~\cite{below_all_subsets3}, \textsc{Dominating Set}~\cite{below_all_subsets3, below_all_subsets4}, and \textsc{Min $k$-Cut}~\cite{min_k_cut}. To some extent, these speedups have been achieved mainly for what we will refer to as \emph{unweighted} problems.
The \emph{weighted} versions of many of these problems are still stuck with the textbook solutions that accept only modest, mostly polynomial-time, improvements.

A concrete example is that of the Traveling Salesman Problem (TSP), which asks to find the shortest tour through $n$ cities.
The textbook algorithm in $O(2^n n^2)$-time remains, up to some modest polynomial improvements~\cite{chan_shaving_logs}, the best up to date for the general setting~\cite{bellman_dp,held_karp_dp, nederlof_tsp}.
In contrast, its unweighted counterpart, the \textsc{Hamiltonicity} problem, has been (non-trivially) sped up to $O^*(1.66^n)$-time by Bj\"orklund~\cite{hamiltonicity}.\footnote{The $O^*$-notation hides polynomial factors in the input size.}

Interestingly, virtually every new result on the unweighted version of a weighted problem has its own theorem stating that the original problem can be solved in the time of its unweighted counterpart,
\emph{yet} with an extra multiplicative pseudo-polynomial dependence on the largest input weight; this is known as the ``(polynomial) embedding technique''.
Notably, this situation also concerns TSP: Bj\"orklund proved that \textsc{Hamiltonicity} can be leveraged to solve TSP by paying the price of a near-linear time-factor depending on the largest edge weight $W$~\cite[Thm.~3]{hamiltonicity}. This becomes impractical for arbitrarily large~$W$.
The same holds true for the weighted \textsc{Max-Cut} problem:~its unweighted counterpart can be solved in $O^*(2^{\omega n/3})$-time~\cite{maxcut-williams, koivisto-correct}, where $\omega < 2.371339$ is the matrix multiplication exponent~\cite{curr-omega}, while its weighted version has to ``accept'' the additional pseudo-polynomial factor. On the other hand, if the weights are polynomially bounded, then the problematic factor becomes polynomial, and the running time of the weighted version matches, in the $O^*$-sense, that of the unweighted version.

\paragraph*{Research Question.} Given the ubiquity of the embedding technique, it is natural to ask whether the gap between the running times of the weighted and unweighted versions of a problem is overcome \emph{only} in the regime of polynomially bounded weights.
A negative answer would be consistent with the widespread adoption of the
embedding approach, with its inherent pseudo-polynomial dependence~\cite{maxcut-williams, koivisto-correct, fsc, set_partitioning, hamiltonicity, saving_space, Cygan_book}.
In contrast, a positive answer would re-open the door to a more systematic re-examination of the weighted vs.\ unweighted gap, in the spirit of the ``Losing Weights by Gaining Edges'' technique for node-weighted problems~\cite{losing-ws} or in terms of approximation guarantees~\cite{apx-ws}.

\subsection{Our Results}

Somewhat intriguingly, we show in this work that the polynomially bounded weights setting is \emph{not} the only one that closes the gap between the weighted and unweighted versions of a problem.
The meta-algorithm we introduce applies to a wide spectrum of \NP-hard combinatorial problems that operate over the $(\min, +)$ and $(\max, +)$ semirings.
Namely, it takes as input a weighted problem instance which guarantees that its weight set has \emph{small doubling}, and solves it in time proportional to that of its unweighted counterpart. The set of input weights $A$ is said to have small doubling if $|A + A| \leq \DC |A|$ for some constant $\DC > 1$, referred to as the \emph{doubling constant}.

To maintain notation consistency with prior work on \emph{unweighted} problems~\cite{esa-paper}, we prefix problem names with ``$\DC$-'' if the set of input weights has doubling constant $\DC$. We also write $\OC$ to note that we suppressed factors that depend only on $\DC$. Under this notation, all of our results conform to the following template:
\begin{restatable}[]{corollary}{DCTSP}
    \label{thm:dc_tsp}
    If \textsc{Hamiltonicity} can be solved in time $T(n)$, then $\DC$-\textsc{TSP} can be solved in time $\OC^*(T(n))$.
\end{restatable}
The applicability of our work goes beyond TSP. Indeed, we can consider any weighted \NP-hard problem that operates over the $(\min,+)$ and $(\max,+)$ semirings, such as the \textsc{Weighted Max-Cut}, \textsc{Edge-Weighted $k$-Clique}, or \textsc{Minimum Steiner Tree} problems.\footnote{We kindly refer the reader to the relevant section for the problem definitions.} The only requirement is that the problem at hand satisfies a property~$\phi$:~its objective value is an additive combination of input weights, and it admits an (algebraic) algorithm as seen earlier; see Prop.~\ref{prop:phi} for a formal treatment.

To avoid redundancy across different problems, we present our results under the umbrella of a meta-algorithm. Namely, given a weighted problem $\mathbf{P}_w$ that satisfies the above property~$\phi$, its small-doubling counterpart can be solved in time matching that of the unweighted version:

\begin{restatable}[Meta-algorithm]{theorem}{METATHM}
    \label{thm:meta}
    If problem $\mathbf{P}_w$ satisfies property $\phi$ and the unweighted version can be solved in time $O(T(n))$ by an algebraic algorithm $\A$, then $\DC$-$\mathbf{P}_w$ can be solved in time $\OC^*(T(n))$.
\end{restatable}

The following corollaries follow:

\begin{restatable}[]{corollary}{DCMAXCUT}
    \label{thm:dc_maxcut}
    If unweighted \textsc{Max-Cut} can be solved in time $T(n, m)$, then $\DC$-\textsc{Weighted Max-Cut} can be solved in time $\OC^*(T(n, m))$.
\end{restatable}
\begin{restatable}[]{corollary}{DCEWCLIQUE}
    \label{thm:dc_ewclique}
    If $k$-\textsc{Clique} can be solved in time $T(n,m,k)$, then $\DC$-\textsc{Edge-Weighted $k$-Clique} can be solved in time $\OC^*(T(n, m, k))$.
\end{restatable}
\begin{restatable}[]{corollary}{DCSTEINER}
    \label{thm:dc_steiner}
    If the unweighted version of \textsc{Minimum Steiner Tree} can be solved in time $T(n, m, k)$, then $\DC$-\textsc{Minimum Steiner Tree} can be solved in time $\OC^*(T(n, m, k))$.
\end{restatable}

We outline in the following how to obtain the results.
Throughout the paper, we refer to the algebraic algorithm that exhibits the pseudo-polynomial dependence on $W$ as the \emph{bounded-input} algorithm.

\subsection{Overview of Approach}\label{sec:overview}

Let us provide some intuition of our approach before we continue with the general framework.

\paragraph*{Intuition.} Consider TSP as a running example. Recall that the bounded-input algorithm works fine as long as $W$ is small and the weights are in $[0, W]$. Intuitively, we can ``hack'' the algorithm to also work with weights that are a subset of a \emph{simple} arithmetic progression $\{V, V + r, \ldots, V + W \cdot r\}$, particularly for large values of $V$ and $r$, namely: Map the weights to $[0, W]$ by subtracting $V$ and dividing by $r$, compute the optimal tour as before on $[0, W]$, and, finally, recover the original objective by scaling by $r$ and adding $nV$. Thus, as long as the arithmetic progression is small, the algorithm still works. Our meta-algorithm lifts this idea to weights lying in a small \emph{generalized} arithmetic progression---a structure that captures small-doubling sets---as follows.

\paragraph*{Meta-Algorithm.} Let $A$ denote the set of input weights of the problem at hand, such as the edge weights in TSP. The first observation is that the objective value of the problem is contained in a $q$-fold sumset of $A$, i.e., $A + \ldots + A$ with $A$ occurring $q$ times, where $q$ is problem-specific; for TSP we have $q = n$. The second observation is that if $A$ has small doubling, then the cardinality of this very $q$-fold sumset is polynomial w.r.t.~$\OC$.\footnote{Noteworthy, Pl\"unnecke's inequality is too loose in this setting, leading to an exponential bound.} This suggests that, \emph{in principle}, we could run the bounded-input algorithm. However, the values may still be numerically as large as $\Theta(q \cdot \max A)$. The missing piece is a decomposition of the weights into constant-size coefficient-tuples coming from the generalized arithmetic progression $A$ is a subset of (cf.~Freiman's theorem~\cite{freiman1}).

To be able to apply this to any bounded-input algorithm out of the box, we must ensure that the coefficient representation preserves the total order of the original weights. We ensure this by designing an order-preserving monomorphism from the coefficient-tuple space into $\{0, \ldots, |A'| - 1\}$, where $q A \subseteq A'$ and $|A'|$ remains polynomial w.r.t.~$\OC$.

\subsection{Related Work}\label{sec:related_work}

We survey applications of additive combinatorics to algorithm design and the literature on the unweighted-weighted dichotomy, focusing on the polynomial embedding technique.

\paragraph*{Additive Combinatorics in Algorithm Design.} Our work continues the effort to consider important results from additive combinatorics~\cite{tao_vu} as a new toolbox for algorithm design, a connection highlighted in works such as Ref.~\cite{trevisan_ac, ac_cs_view3, ac_cs_view1, ac_cs_view2}.
One of the first seminal results is that of Chan and Lewenstein~\cite{chan-lewenstein}, who used the (constructive) Balog-Szemer\'edi-Gowers (BSG) theorem to speed up the monotone variant of the min-plus convolution on small integers. Recently, this has been further used and extended in \#APSP~\cite{fredman_meets_domi}, 3SUM~\cite{3sum_sidon}, listing 4-cycles~\cite{abf}, and approximate distance oracles~\cite{abf}. We observe that the constructive version of a fundamental theorem in additive combinatorics bears fruit in algorithm design. To this end, Randolph and W\k{e}grzycki~\cite{esa-paper} made Freiman's theorem~\cite{freiman1}, another fundamental result in additive combinatorics, constructive. Their main applications were to \textsc{SubsetSum} and \textsc{ILP-Feasibility}, which have an inherent additive structure.
Indeed, we extend their results to the area of weighted problems.

\paragraph*{Embedding Technique.} Repurposing the unweighted algorithm for the weighted version of a problem is a common theme in work that improves over the textbook solution,
as in seminal results such as Ref.~\cite{maxcut-williams, fsc, set_partitioning, hamiltonicity}; notably, in all these examples, the weighted problem operates over the $(\min, +)$ and $(\max, +)$ semirings.
The key idea is to represent the input weights of the problem instance as monomials so that the bounded-input algorithm can work in the $(+, \times)$ ring.
The drawback is that this introduces a pseudo-polynomial dependence in the running time,
since operations now incur an additional $\widetilde O(W)$ factor,\footnote{The $\widetilde O$-notation hides polylogarithmic factors in the input size.} where $W$ is the largest input weight.
For instance, undirected \textsc{Hamiltonicity} runs in $O^*(1.66^n)$-time, hence TSP can be solved in $O^*(1.66^n W)$-time. This, however, becomes impractical for arbitrarily large $W$. Still, if $W$ is polynomially bounded, then up to polynomial factors, the weighted problem becomes equivalent to the unweighted one. For example, if $W = \poly(n)$, then \textsc{TSP} can be solved in time $O^*(1.66^n)$, matching its unweighted counterpart.

\paragraph*{Dodging Pseudo-Polynomiality.} The pseudo-polynomial factor is often unwanted. One way to shave it is to enter the realm of approximation. In fact, there is a simple exercise to turn any bounded-input algorithm for a problem into an exponential-time $(1+\varepsilon)$-approximation algorithm. For the \textsc{Max-Cut} problem, Williams mentions this idea in passing~\cite{maxcut-williams}, inspired by the literature on the all-pairs shortest paths problem~\cite{zwick_apsp}; see, e.g., Ref.~\cite{approx-min-sum-conv} for more such examples. It is interesting to ask whether the doubling constant could play a role in approximation algorithms.

\subsection{Roadmap}

We structure the paper as follows. We start with preliminaries on the embedding technique (Sec.~\ref{subsec:embedding}) and the necessary notation in additive combinatorics (Sec.~\ref{subsec:ac}). Then, we introduce the meta-algorithm (Sec.~\ref{sec:meta-algo}), including the order-preserving monomorphism (Sec.~\ref{sec:kappa}). Subsequently, we instantiate the meta-algorithm on several weighted problems (Sec.~\ref{sec:apps}). We conclude with an outlook on polynomial-time weighted problems (Sec.~\ref{sec:outlook}).

\section{Preliminaries}

\paragraph*{Notation.} For $a < b \in \mathbb{Z}$, define $[a, b] \triangleq \{a, \ldots, b\}$. As shorthand, we define $[n] \triangleq [1, n]$. Moreover, we define $f(X)$ as $\{f(x) \mid x \in X\}$ for a function $f$ and set $X$.

\paragraph*{Time Complexities.} As we will only deal with exponential-time algorithms, we use the $O^*$-notation to hide polynomial factors in the input size.
For example, $O^*(2^n) = O(2^n \poly(n))$.
In particular, this notation will never hide pseudo-polynomial factors depending on the input weights; specifically, we mean the largest input weight, which will be denoted by $W$.
Furthermore, we use the $\OC$-notation, already established in Ref.~\cite{esa-paper}, to denote that we suppressed factors depending on $\DC$, the doubling constant on which we parameterize the time-complexities.
For instance, $\OC(2^n) = f(\DC) \cdot O(2^n)$ for a computable function $f$.
To put into context, $\OC^*(2^n)$ will thus mean $f(\DC) \cdot O^*(2^n)$.

\paragraph*{Graphs.} Unless otherwise specified, all problems we consider in this work operate on an undirected graph $G = (V, E)$, most of the time endowed with a weight function $w$ defined on the edge set $E(G)$ that attributes each edge $e$ a positive value.
The length of a path $\pi = (v_1, \ldots, v_k)$ is $k - 1$, i.e., we count the number of edges, while the weight of the path is the sum of all edge weights, i.e., $w(\pi) \vcentcolon= \sum_{i=1}^{k-1} w(v_i, v_{i+1})$.

\subsection{Embedding Technique}\label{subsec:embedding}

We present the embedding technique, which is used by several works to leverage the running time of the unweighted problem for the corresponding weighted counterpart, at the cost of a pseudo-polynomial dependence on the largest input weight, $W$. Specifically, we consider \NP-hard problems that admit formulations over the $(\min, +)$ and $(\max, +)$ semirings: combining substructures corresponds to addition, while selecting among candidate values corresponds to taking a minimum (for minimization) / maximum (for maximization).

\paragraph*{Overview.} The key idea is to encode weights as exponents in a polynomial ring, since (i) the ``$+$'' operator of our semirings translates to multiplication between monomials ($x^{a+b}=x^{a}\cdot x^{b}$), and (ii) aggregating candidate values corresponds to polynomial addition, which retains information about all previously attainable values.\footnote{If the reader will, this operation ``simulates'' an additive inverse.} Hence, the output of the bounded-input algorithm is a \emph{solution polynomial}:

\begin{definition}[Solution polynomial]\label{def:sol-poly}
    Let $I_w$ be an instance of a weighted combinatorial problem with a finite set $\mathcal{F}(I_w)$ of feasible solutions. For each $\sigma\in\mathcal{F}(I_w)$, let $v(\sigma)\in\mathbb{Z}_{\ge 0}$ denote its objective value. The \emph{solution polynomial} of $I_w$ is the ordinary generating function
    \begin{equation*}
        F_{I_w}(x) = \sum_{\sigma \in \mathcal{F}(I_w)} x^{v(\sigma)} = \sum_{s = 0}^{V_{\max}} m_{I_w}(s) \, x^{s},
    \end{equation*}
    where $m_{I_w}(s)\coloneqq \bigl|\{\sigma\in\mathcal{F}(I_w) \mid v(\sigma)=s\}\bigr|$ and $V_{\max}\coloneqq \max_{\sigma\in\mathcal{F}(I_w)} v(\sigma)$.
\end{definition}
The solution polynomial abstracts away whether the problem is a minimization or maximization problem: for minimization we select the monomial with the smallest exponent, whereas for maximization we select the monomial with the largest exponent. To this end, we refer to an algorithm as \emph{algebraic}, if it returns a solution polynomial for an instance of a problem.\footnote{An alternative attribute could be \emph{non-combinatorial}. For a treatment on the definition of a combinatorial algorithm, see Ref.~\cite{fast-bmm}.}

Due to the additive structure of the input weights we will be exploiting, starting in Sec.~\ref{sec:meta-algo}, we will have to work with coordinate tuples instead of the standard setting of integer weights. To overcome this, we will propose a pairing function that transforms tuples into equivalent integer weights; see Sec.~\ref{sec:kappa} for the technical details.

\subsection{Additive Combinatorics}\label{subsec:ac}

As the trend of importing results from additive combinatorics into algorithm design is itself rather young,
we provide the necessary notation so that our algorithms can be easily followed.

\paragraph*{Basic Notation.} Let $A$ and $B$ be two sets.
Define their sumset (also known as the Minkowski sum) $A + B$ in a commutative group as $\{a + b \mid a\in A, b\in B\}$.
When referring to $A + A$, we directly note $2A$. Hence, the iterated sumset $h$-fold $hA$ is defined recursively by $hA = (h - 1)A + A$.

\paragraph*{Doubling Constant.} An important definition is that of the doubling constant of a set $A$, which quantifies how large the sumset becomes w.r.t.~to the original set:
\begin{equation}
    \DC(A) = \frac{|A + A|}{|A|}.
    \label{eq:dc}
\end{equation}
Whenever $\DC(A)$ does not depend on the size of $A$, we directly write $\DC$. We say that these sets have \emph{small doubling}.
Our work shows that whenever the set of input weights of \NP-hard problems has small doubling, that instance can be solved faster.

This comes as a surprise given the simple definition of the doubling constant in Eq.~\eqref{eq:dc}.
To gain an intuitive understanding, consider the case when $A$ is a (simple) arithmetic progression, say $\{2, 4, 6, 8\}$.
Its 2-fold sumset, $A + A$, is $\{4, 6, 8, 10, 12, 14, 16\}$.
Taking a more general example will lead us to conclude that the sumset of a (simple) arithmetic progression does not explode that much.
Indeed, $|A + A| = 2|A| - 1$, whenever $A$ is a (simple) arithmetic progression.
On the contrary, the set $\{3, 5, 9, 17\}$ results in a sumset that attains the maximum size for a 4-size set, namely 10 elements.
In other words, there are no distinct pairs in $A \times A$ that have the same sum.

\paragraph*{GAPs.} Apart from the well-known \emph{simple} arithmetic progression, additive combinatorics results target \emph{generalized} arithmetic progressions (GAP).
Formally, a GAP $G$ is defined as \[
    G = \{x_1\ell_1 + \ldots + x_d\ell_d \mid \forall i\in[d], \ell_i \in [0, L_i]\},
\]
where $d$ is the dimension of the GAP, $\{x_1, \ldots, x_d\}$ are the generators, and $\{L_1, \ldots, L_d\}$ are the dimension bounds; in particular, a 1-dimensional GAP is a simple arithmetic progression.

In our previous examples, we saw that if $A$ \emph{is} an arithmetic progression, then it has a small doubling constant;
the interesting problems in additive combinatorics are indeed the \emph{inverse} problems:~\emph{What does a set of small doubling look like?} This is where Freiman's theorem comes into play.

\paragraph*{Freiman's Theorem.} Given a set $A$ of small doubling, Freiman's theorem~\cite{freiman1} tells that, indeed, $A$ is included in a GAP $G$, whose dimension is independent on $|A|$.
In this work, we will need the \emph{constructive} Freiman's theorem, recently designed by Randolph and W\k{e}grzycki~\cite{esa-paper}:
\begin{theorem}[Constructive Freiman's Theorem~\cite{esa-paper}]
    \label{thm:con_freiman}
    Let $A$ be a set of $n$ integers with $|A + A| \leq \DC |A|$. Then, there exists an $\tilde{O}_\DC(n)$-time algorithm that with probability $1 - n^{-\gamma}$, for an arbitrarily large constant $\gamma > 0$, returns a generalized arithmetic progression
         \[
             G = \{x_1\ell_1 + x_2\ell_2 + \dots + x_{d(\DC)}\ell_{d(\DC)} \mid \forall i, \ell_i \in [L_i] \} \supseteq A
         \]
    with dimension $d(\DC)$ and volume $v(\DC)|A|$, where $d$ and $v$ are computable functions that depend
    only on the doubling constant $\DC$. Specifically, the theorem provides the values $x_1, x_2, \dots, x_{d(\DC)}$ and $L_1, L_2, \dots L_{d(\DC)}$. 
\end{theorem}
The above theorem gives us a GAP of dimension $d(\DC) = 2^{\DC^{O(1)}}$ and volume $v(\DC) = 2^{2^{\DC^{O(1)}}} n$, as in the original Freiman's theorem~\cite{freiman1}. Noteworthy, later works do optimize $d$ and $v$~\cite{freiman-better-1, freiman-better-2, freiman-better-3, freiman-better-4}, yet these improvements are not captured by the constructive version.

This completes the necessary preliminaries on additive combinatorics.
We are now ready to show that, in the $\OC$-sense, \NP-hard weighted problems reduce to their unweighted versions whenever their input weights form a set of small doubling.

\section{Meta-Algorithm}\label{sec:meta-algo}

Before outlining the meta-algorithm, we first define the property that the considered weighted problems have to satisfy to be amenable to the recent speedups achieved for their unweighted counterparts. As motivated in the introduction, the key requirement is that the objective value has to be an additive combination of input weights. This additive structure ensures that objective values lie within a folded sumset of the input weights.

\begin{property}[Property $\phi$]
\label{prop:phi}
Let $\mathbf{P}_w$ be a weighted problem and $\mathbf{P}$ its unweighted counterpart. Let $I$ be an instance of $\mathbf{P}_w$ of size $n$, $w(I)$ its set of weights, and $W = \max w(I)$. Then, we say that $\mathbf{P}_w$ has property $\phi$ if the following hold:
\begin{enumerate}
    % Objective value characterization.
    \item Any feasible solution to $I$ is the total weight of a set $S \subseteq w(I)$,
    \label{cond:solution}

    % Algebraic algorithm.
    \item There is an algebraic algorithm $\mathcal{A}$ that solves $\mathbf{P}$ in $O(T(n))$-time and $\mathbf{P}_w$ in $O^*(T(n) \cdot W)$-time, and any intermediate solution to $I$ produced by $\mathcal{A}$ is the total weight of a polynomial-size multi-set with support in $w(I)$. \label{cond:algo}
\end{enumerate}
\end{property}

Notably, the additional requirement in Condition~\ref{cond:algo} is inherent to any bounded-input algorithm whose running time has the near-linear dependence on the largest input weight $W$. Next, we introduce the pairing function that maps GAP coordinates to integers, allowing the bounded-input algorithm to operate directly on integers rather than raw coordinate tuples.

\subsection{Pairing Function}\label{sec:kappa}

To be able to actually run the bounded-input algorithm $\mathcal{A}$, we need to transform GAP coordinates into integers. In particular, the meta-algorithm will first enlarge the dimension bounds before encoding them as integers, so that the coordinates set is \emph{well-defined}, namely:~Given any two GAP coordinate tuples $\alpha = \tup{\alpha_1, \ldots, \alpha_d}$ and $\beta = \tup{\beta_1, \ldots, \beta_d}$ of dimension $d$, it holds that $\alpha_i + \beta_i \leq L_i$, $\forall i \in [d]$.

In addition, as we operate over the $(\min, +)$ and $(\max, +)$ semirings, the pairing function $\kappa$ has to preserve the ``additivity'' of the coordinates, namely:\footnote{Note that we skip redundant information from $\kappa$, such as the dimension and the dimension bounds.}
\begin{equation}
    \kappa(\alpha \oplus \beta) = \kappa(\alpha) + \kappa(\beta),
    \label{eq:kapp_prop}
\end{equation}
where ``$\oplus$'' is the entrywise addition operator, i.e., $\tup{\alpha \oplus \beta}_i \vcentcolon= \alpha_i + \beta_i, \forall i \in [d]$. Since $\kappa$ must be injective, so that we can uniquely get back the GAP coordinates, $\kappa$ has to be a monomorphism between a well-defined set of GAP coordinates and the encoded integers.

\paragraph*{Our Pairing Function.} We propose the following pairing function:
\begin{equation*}
    \kappa\left(d, \tup{L_i}_{i \in [d]}, \tup{l_i}_{i \in [d]}\right)=
    \begin{cases}
        l_1, & \text{if } d = 1, \\
        l_d + (L_d + 1) \cdot \kappa\left(d - 1, \tup{L_i}_{i \in [d-1]}, \tup{l_i}_{i \in [d-1]}\right)\!, & \text{if } d > 1.
    \end{cases}
\end{equation*}
It receives the dimension of the GAP, $d$, its dimension bounds $\tup{L_i}_{i \in [d]}$, and the GAP coordinates $\tup{l_i}_{i \in [d]}$ of an element, and returns the encoded value.
If there is only one dimension, then the corresponding GAP coordinate is immediately returned.
Otherwise, it builds the encoded value recursively, by multiplying the result for the remaining $d - 1$ dimensions by $L_d + 1$ and finally adding $l_d$.
We show that in our setting it holds that $\kappa$ is a monomorphism:

\begin{restatable}[]{lemma}{PAIRING}
    \label{lemma:pairing}
    Let $\tup{L_i}_{i \in [d]} \in \mathbb{N}^d$ be the bounds of a GAP of dimension $d$. Given a well-defined set $D \subseteq \prod_{i=1}^{d} [0, L_i]$ of $d$-size coordinate tuples, the pairing function $\kappa$ is a monomorphism between $(D, \oplus)$ and $([0, \prod_{i=1}^{d}(L_i + 1) - 1], +)$.
\end{restatable}
We defer the proof of Lemma~\ref{lemma:pairing} to Appendix~\ref{appendix:mono}. As a by-product of our injectivity proof, we show that $\kappa$ satisfies a certain monotonicity property:
\begin{equation}
    \label{eq:nice_monotonicity}
    \alpha\prec\beta \implies \kappa(\alpha) < \kappa(\beta),
\end{equation}
where $\alpha \prec \beta$ stands for the (syntactic) lexicographical order of the coordinate tuples, e.g., $\tup{1, 2} \prec \tup{1, 3}, \tup{1, 2} \prec \tup{2, 1}$. However, even then, there is still an issue:~\emph{monomorphism alone does not suffice}. We discuss why this is the case in the following.

\paragraph*{Why Monomorphism Alone Is Not Sufficient.} The present form of $\kappa$ does not suffice to show the correctness of our meta-algorithm. The salient issue is that $\kappa$ does \emph{not} preserve the total order of the tuples regarding their original, to-be-represented value.
Consider for instance the coordinate tuples $\tup{1, 2}$ and $\tup{2, 1}$.
While we know that $\tup{1, 2} \prec \tup{2, 1}$ in the lexicographical sense and implicitly $\kappa(\tup{1,2}) < \kappa(\tup{2,1})$, according to Eq.~\eqref{eq:nice_monotonicity}, it could be that the generators of the GAP are $\tup{3, 10}$.
In that case, the tuples should be ordered as $\tup{2, 1} \sqsubset \tup{1, 2}$, since $3 \cdot \emph{2} + 10 \cdot \emph{1} < 3 \cdot \emph{1} + 10 \cdot \emph{2}$, where we refer to ``$\sqsubset$'' as the total order on the GAP coordinates.

\paragraph*{Restoring the Order.} To address this, we augment the call to the bounded-input algorithm with a (fixed) \emph{permutation} $\pi$ that restores the order of the original values. Whenever the algorithm has to find the minimum~/~maximum monomial in the solution polynomial (Def.~\ref{def:sol-poly})---the final last step in the bounded-input algorithms we consider---, it uses $\pi$ to recover the true rank of the current encoding $e$.

\paragraph*{Building $\pi$.} The permutation can be built in polynomial-time w.r.t.~$\OC$, as follows: Once the dimension bounds of the GAP have been enlarged (see line~\ref{line:3} in the upcoming Alg.~\ref{algo:meta}), we can generate the values $\sum_{i=1}^{d} x_i \cdot l_i$ for $l_i \in [0, \lam L_i], \forall i \in [d]$, where $\lam$ is problem-specific (see the upcoming Sec.~\ref{sec:exploiting}). By construction, these are all possible values that can be taken during a bounded-input algorithm's run. To actually ``freeze'' the permutation $\pi$, we have to sort the generated values, which takes time $\widetilde O(\lam ^d \prod_{i=1}^{d} (L_i + 1))$. We call this step $\textsc{BuildPermutation}$ in the meta-algorithm.

\paragraph*{Decoding.} The decoding process follows a similar line and is outlined in the following:
\begin{equation*}
    \kappa^{-1}\left(d, \tup{L_i}_{i \in [d]}, e\right)=
    \begin{cases}
        \tup{e}, & \text{if } d = 1, \\
        \kappa^{-1}\left(d - 1, \tup{L_i}_{i \in [d-1]}, \left\lfloor\frac{e}{L_{d}+1}\right\rfloor\right) \circ \tup{e\bmod{(L_d + 1)}}, & \text{if } d > 1.
    \end{cases}
\end{equation*}
Namely, if there is only one dimension, the current encoding is returned, otherwise we append the result of $e \bmod{(L_d + 1)}$ to the tuple resulting from the remaining $d - 1$ dimensions. The recursive call is done with an $e$ from which the contribution of dimension $d$ has been removed.

\subsection{Exploiting the Bounded-Input Algorithm}\label{sec:exploiting}

Let $\mathbf{P}_w$ be a weighted problem satisfying property~$\phi$, and let $I$ be an instance of $\mathbf{P}_w$ of size~$n$ whose set of weights has small doubling. Let $\mathcal{A}$ be the bounded-input algorithm from Condition~\ref{cond:algo}. Furthermore, let $\lambda = \poly(|w(I)|)$ be an upper-bound on the cardinality of the multi-sets from the same condition; as we will see, for all bounded-algorithms we consider in this paper it holds that $\lambda = O(|w(I)|)$. We can now state our meta-algorithm, outlined in Alg.~\ref{algo:meta}.

\begin{restatable}[]{algorithm}{META_ALGO}
    \caption{$\mathcal{M}(I, \mathcal{A})$}
	\label{algo:meta}
\begin{algorithmic}[1]
    \State $d(\DC), \tup{x_1, \ldots, x_{d(\DC)}}, \tup{L_1, \ldots, L_{d(\DC)}} \gets \textsc{ConstructiveFreiman}(w(I))$ [Thm.~\ref{thm:con_freiman}] \label{line:1}
    \State $w_{\text{coord}} \gets \textsc{GetGapCoordinates}\left(
        w(I),
        \tup{x_1, \ldots, x_{d(\DC)}}, \tup{L_1, \ldots, L_{d(\DC)}}
    \right)$ \label{line:2}
    \State $w' \gets \kappa\left(d(\DC), \tup{\lam L_1, \ldots, \lam L_{d(\DC)}},  w_{\text{coord}}\right)$ \label{line:3}
    \State $\pi \gets \textsc{BuildPermutation}\left(
        d(\DC),
        \tup{x_1, \ldots, x_{d(\DC)}},
        \tup{\lam L_1, \ldots, \lam L_{d(\DC)}},
        w_{\text{coord}}
    \right)$ \label{line:pi}
    \State $c' \gets \mathcal{A}(w', \pi)$ \label{line:4}
    \State $\tup{l_1', \ldots, l_{d(\DC)}'} \gets \kappa^{-1}\left(d(\DC), \tup{\lam L_1, \ldots, \lam L_{d(\DC)}}, c'\right)$ \label{line:5}
    \State \Return $\sum_{i=1}^{d(\DC)} x_i l_i'\:(=\vcentcolon c^*)$ \label{line:6}
\end{algorithmic}
\end{restatable}

\METATHM*
\begin{proof}
    We prove that the bounded-input algorithm actually receives as input polynomially bounded weights w.r.t.~$\OC$ via the new weight function $w'$ (line~\ref{line:4}).
    This will guarantee that the line runs in time $\OC^*(T(n))$. Note that this does not hide pseudo-polynomial factors in the largest input weight.
    
    To this end, let $G'$ be the GAP which $\lam w(I)$, a superset of any intermediate solution---recall Condition~\ref{cond:algo}, is a subset of.
    Moreover, the initial set of input weights $w(I)$ is also a subset of a GAP $G$. Let us see the connection between $G'$ and the original $G$.
    First, the parameters of $G$ are $\{x_1, \ldots, x_{d(\DC)}\}$ and $\{L_1, \ldots, L_{d(\DC)}\}$, i.e., \[
        G = \{x_1 l_1 + \ldots + x_{d(\DC)} l_{d(\DC)} \mid \forall i, l_i \in [L_i]\}.
    \]
    Then, since $G'$ has to capture the weights in $\lam w(I)$, we need to modify the bounds of the original GAP.
    We thus obtain the following new GAP:\[
        G' = \{x_1 l_1 + \ldots + x_{d(\DC)} l_{d(\DC)} \mid \forall i, l_i \in [\lam L_i]\}.
    \]
    Recall now that the coordinates are mapped to the integer range $[0, \prod_{i=1}^{d(\DC)} (\lambda L_i + 1) - 1]$, which is exactly the volume of $G'$. Bounding this, we obtain:
    \begin{align}
    |G'| = \prod_{i=1}^{d(\DC)} (\lambda L_i + 1)
        &\leq \lambda^{d(\DC)} \prod_{i=1}^{d(\DC)} (L_i + 1) \nonumber \\
        &= \lambda^{d(\DC)} |G| \nonumber \\
        &= \lambda^{d(\DC)} v(\DC)\,|w(I)| \nonumber \\
        &= |w(I)|^{\OC(1)}. \label{eq:ineq}
    \end{align}
    Therefore, the bounded-input algorithm (line~\ref{line:4}) is run on an instance with polynomially large weights w.r.t.~$\OC$, represented by the weight function $w'$.
    
    Let us analyze the total running time.
    We argue that all steps apart from line~\ref{line:4} take polynomial time w.r.t.~$\OC$ each.
    Namely, once the parameters of the GAP $G$ have been computed (line~\ref{line:1}),
    we can compute the coordinates corresponding to the values of the weight function $w$ by a naive iteration of $G$ (line~\ref{line:2}); this takes time $O(|w(I)| \cdot |G|) = O(|w(I)| \cdot v(\DC) |w(I)|) = \OC(|w(I)|^2)$. Then, building the permutation $\pi$ (line~\ref{line:pi}) takes time $\widetilde O(|G'|) = \widetilde O(|w(I)|^{O_\DC(1)})$; see Sec.~\ref{sec:kappa} and the above Eq.~\eqref{eq:ineq}.
    Finally, applying $\kappa$ on the GAP coordinates $w_{\text{coord}}$ takes $\OC(|w(I)|)$-time (line \ref{line:3}), while lines \ref{line:5} and \ref{line:6} take $\OC(1)$-time and $\widetilde{O}_{\DC}(1)$-time, respectively.

    Noteworthy, when not parameterizing on the doubling constant, the running time of the meta-algorithm reads $O^*(\lambda^{d(\DC)} v(\DC) T(n))$, which is $\OC^*(T(n))$.
\end{proof}

\section{Applications}\label{sec:apps}

In the following, we review several weighted \NP-hard problems and show how to solve their small-doubling instances faster than the traditional algorithm by leveraging the previous meta-algorithm. This enables running times in the regime of the unweighted version of the problem for a large suite of problems, such as \textsc{TSP}, \textsc{Weighted Max-Cut}, \textsc{Edge-Weighted $k$-Clique}, and \textsc{Minimum Steiner Tree}. We stress that, compared to prior work~\cite{esa-paper}, we consider \emph{weighted} problems that operate over tropical semirings.

\subsection{TSP Meets Small Doubling}\label{subsec:tsp-dc}

We define the doubling constant parametrization of TSP, $\DC$-TSP, as follows:

\defproblem{$\mathbf{\mathcal{C}}$-\textbf{TSP}}
{Undirected graph $G = (V, E)$, a cost $w(e) \in \mathbb{Z}_{\geq 0}$ for each $e \in E = E(G)$, such that $|w(E) + w(E)| \leq \DC |w(E)|$}
{Tour $T \subseteq E$ that minimizes $\displaystyle\sum_{e\in T} w(e)$}

\paragraph*{Why the Embedding Technique Fails.} First, note that bounded-input algorithm does not help much for $\DC$-\textsc{TSP}: The largest edge weight $W$ can still be exponentially large. Therefore, the resulting running time would be even worse than that of the traditional algorithm. We show an alternative solution via our meta-algorithm (Alg.~\ref{algo:meta}).

\paragraph*{Intuition.} Let $A \vcentcolon= w(E)$ denote the set of edge weights. Recall that the weight of any tour in TSP is computed by the total sum of the edge weights on that tour. To understand how this actually relates to the doubling constant, consider the case when we only consider \emph{paths} of length 2. The possible weights of these paths are all included in $A + A$. Note that not all path weights are also valid: some weights might come from paths using the same edge twice. This, however, is not an issue. The main benefit we get is that the cardinality of $A + A$ is at most $\DC |A|$. If we continue with longer paths, we will observe the same situation: Namely, the weights of the paths of length $k$ will be included in $k A$, the $k$-fold sumset of $A$. In particular, the shortest tour length will be an element of $nA$. 

\begin{lemma}
    TSP satisfies property $\phi$.
\end{lemma}
\begin{proof}
    Condition~\ref{cond:solution} holds since a tour’s weight is the sum of the weights of its constituent edges. Next, we can take $\mathcal{A}$ as Bj\"orklund's algorithm~\cite{hamiltonicity} for undirected \textsc{Hamiltonicity}, which runs in time $O^*(1.66^n)$, and solves the weighted case in time $O^*(1.66^n\,W)$~\cite[Thm.~3]{hamiltonicity}. Due to the algorithm's design, namely that it counts cycle covers of $n$ arcs, we have $\lambda \leq n$.
\end{proof}

\DCTSP*

\subsection{Weighted Max-Cut Under Small Doubling}\label{sec:maxcut}

Another well-studied problem whose unweighted version has been sped up and fits the small doubling setting well is the \textsc{Weighted Max-Cut} problem.
In a breakthrough result, Williams~\cite{maxcut-williams} improved the time complexity of the unweighted \textsc{Max-Cut} to $O^*(2^{\omega n / 3})$.
We show that we can leverage this result to the weighted setting when the edge weights form an integer set that has small doubling.
We formally define the doubling constant parametrization of the problem in the following:

\defproblem{$\mathbf{\mathcal{C}}$-\textbf{Weighted Max Cut}}
{Undirected graph $G = (V, E)$, a weight $w(e) \in \mathbb{Z}_{\geq 0}$ for each $e \in E = E(G)$ such that $|w(E) + w(E)| \leq \DC |w(E)|$} 
{Subset of vertices $S$ so that $\sum_{e\in\delta(S)} w(e)$ is maximized}

\begin{lemma}
    \textsc{Weighted Max-Cut} satisfies property $\phi$.
\end{lemma}
\begin{proof}
    Condition~\ref{cond:solution} holds since a cut’s weight is the sum of the weights of its constituent edges. Next, we can take $\mathcal{A}$ as Williams's algebraic algorithm~\cite{maxcut-williams} (and later improved by Koivisto~\cite{koivisto-correct}) for the unweighted \textsc{Max-Cut} problem, which runs in time $O^*(2^{\omega n / 3})$, and solves the weighted case in time $O^*(2^{\omega n / 3}\,W)$~\cite[Thm.~3]{koivisto-correct}.
\end{proof}

Using the above lemma together with Thm.~\ref{thm:meta}, we obtain:

\DCMAXCUT*

\subsection{Small Doubling and Edge-Weighted $k$-Cliques}

We now turn to a problem whose complexity is a rather popular hypothesis in hardness proofs~\cite{clique-hyp-1, clique-hyp-2, clique-hyp-3}: the edge-weighted $k$-clique problem. Compared to the problems considered so far, the algebraic algorithm for the unweighted $k$-Clique has already been known from the 80's~\cite{uw-k-clique}. We show that it can be leveraged in the doubling constant parametrization of the edge-weighted version. Note that, apart from a mention in passing in Lincoln et al.~\cite{clique-hyp-3} that Nešetřil-Poljak's algorithm can also be used for polynomially bounded weights, we could not find any reference to a proof. For completeness, we show this in Lemma~\ref{lemma:np} (Appendix~\ref{appendix:NP}).

The doubling constant parametrization of the problem is the following:

\defproblem{$\mathbf{\mathcal{C}}$-\textbf{Edge-Weighted $k$-Clique}}
{Undirected graph $G = (V, E)$, integer $k$, and a weight $w(e) \in \mathbb{Z}_{\geq 0}$ for each $e \in E = E(G)$ such that $|w(E) + w(E)| \leq \DC\, |w(E)|$}
{A vertex set $S \subseteq V$ with $|S| = k$ that induces a clique and maximizes $\sum_{e \in E(G[S])} w(e)$}

\begin{lemma}
    \textsc{Edge-Weighted $k$-Clique} satisfies property $\phi$.
\end{lemma}
\begin{proof}
    Condition~\ref{cond:solution} holds since a clique’s weight is the sum of the weights of its constituent edges. Next, we can take $\mathcal{A}$ as Nešetřil-Poljak's algebraic algorithm~\cite{uw-k-clique} for the unweighted version of the problem (or the improvement by Eisenbrand and Grandoni~\cite{faster-uw-k-clique} for some values of $k$ using rectangular matrix multiplication), which runs in time $O^*(n^{k \omega / 3})$, and solves the weighted case in time $O^*(n^{k \omega / 3}\,W)$ (Lemma~\ref{lemma:np}). Notably, due to its design, namely that each original input weight is \emph{doubled} in the auxiliary graph, the bounded-input algorithm (Lemma~\ref{lemma:np}) has a $\lambda \leq 2 |w(I)|$.
\end{proof}

In combination with Thm.~\ref{thm:meta}, we obtain:

\DCEWCLIQUE*

\subsection{Steiner Meets Freiman}\label{sec:min-steiner-tree}

For reference, we also include the \textsc{Minimum Steiner Tree}, even though the general setting can be directly solved in $O((2 + \epsilon)^k\, p(n))$~\cite{faster_steiner1}, where $p(n)$ is a polynomial function of $n$, the degree of which grows rapidly as $\epsilon$ approaches zero, Its degree has been later refined to $12 \sqrt{\epsilon^{-1} \ln \epsilon^{-1}}$, having as a result running times such as $O(2.1^k\,n^{57.6})$ or $O(2.5^k\,n^{14.2})$~\cite{fuchs2007dynamic}. Notably, Dreyfus and Wagner~\cite{dreyfus1971steiner} were the first to design a dynamic programming recursion running in time $O(3^kn + 2^k n^2 + nm)$, where $n = |V|$, $m = |E|$, and $k = |T|$.

The doubling constant parametrization of the problem is the following:

\defproblem{$\mathbf{\mathcal{C}}$-\textbf{Minimum Steiner Tree}}
{Undirected graph $G = (V, E)$, a weight $w(e) \in \mathbb{Z}_{\geq 0}$ for each $e \in E = E(G)$, and a set of vertices $K \subseteq V = V(G)$} 
{Subgraph $H$ of $G$ that connects the vertices in $K$ and has the minimum total weight $\sum_{e\in E(H)} w(e)$ among all such subgraphs of $G$}

Next, we show that the minimum Steiner tree is amenable for speedups in this regime:
\begin{lemma}
    \textsc{Minimum Steiner Tree} satisfies property $\phi$.
\end{lemma}
\begin{proof}
    Condition~\ref{cond:solution} holds since, as the edge weights are all non-negative, the optimal subgraph $H$ is a tree with its leaves in $K$, and a tree’s weight is the sum of the weights of its constituent edges. Next, we can take $\mathcal{A}$ as the fast subset convolution based algorithm due to Bj\"orklund, Husfeldt, Kaski, and Koivisto~\cite[Sec.~4.1.2]{fsc}, which runs in $O^*(2^k n^2 W)$-time, or Lokshtanov--Nederlof's polynomial-space algorithm~\cite{saving_space} in the same running time.
\end{proof}

With this, using Thm.~\ref{thm:meta}, we obtain:

\DCSTEINER*

This concludes the suite of applications of our meta-algorithm.~We conclude with an outlook on the applicability of our framework beyond \NP-hard weighted problems.

\section{Outlook:~Beyond \NP-hard Problems}\label{sec:outlook}

As shown in this paper, the constructive Freiman's theorem can be converted in a rather effective application for \NP-hard weighted problems as well. A natural question arises: Is it possible to extend its applicability to weighted problems in \textsf{P}?

\paragraph*{Min-Plus Convolution.} The answer seems to be a positive one, in particular for those polynomial problems for which the solution is an additive combination of the input weights. A simple application is that of the $(\min, +)$-convolution:~Given sequences $(a[i])_{i=0}^{n - 1}$ and $(b[i])_{i=0}^{n - 1}$, the goal is to compute $(c[i])_{i=0}^{n - 1}$, where $c[k] = \min_{i=0, \ldots, k} \{a[i]\,+\,b[k - i]\}$. The naive algorithm runs in $O(n^2)$-time, and no $O(n^{2 - \varepsilon})$-time algorithm is known for $\varepsilon > 0$~\cite{min-plus-1, min-plus-2}. However, the convolution in the $(+, \times)$-ring can be solved in $O(n \log n)$-time via FFT. This resembles the setting we considered throughout the paper. Indeed, Węgrzycki outlines in his PhD thesis a bounded-input algorithm running in time $\widetilde O(n W)$, where $W$ is the largest input value~\cite[Lemma~5.7.2]{karol_phd}. Applying the same key idea as for the \NP-hard problems we considered, we can obtain an $\widetilde O_\DC(n)$-time algorithm for the min-plus \emph{self}-convolution problem, i.e., $a = b$, and the input sequence, when regarded as a set, has small doubling. The same setting holds true for the min-sum subset convolution~\cite{fsc}.

It is thus interesting to ask how the framework behaves in more weighted applications, such as (specialized) $(\min, +)$ matrix products~\cite{special-min-plus-prod} or even \textsc{APSP}, where---similar to the \NP-hard case studied here---the objective value is an additive combination of input weights.

\paragraph{Acknowledgments.} The author thanks the anonymous reviewers of STACS'25 and STACS'26, whose detailed comments visibly improved the quality of the presentation. 

\bibliography{mind-the-gap}

\newcommand{\etalchar}[1]{$^{#1}$}
\begin{thebibliography}{CMWW17}

\bibitem[ABF23]{abf}
Amir Abboud, Karl Bringmann, and Nick Fischer.
\newblock {Stronger 3-SUM Lower Bounds for Approximate Distance Oracles via Additive Combinatorics}.
\newblock In Barna Saha and Rocco~A. Servedio, editors, {\em Proceedings of the 55th Annual {ACM} Symposium on Theory of Computing, {STOC} 2023, Orlando, FL, USA, June 20-23, 2023}, pages 391--404. {ACM}, 2023.
\newblock \href {https://doi.org/10.1145/3564246.3585240} {\path{doi:10.1145/3564246.3585240}}.

\bibitem[ADV{\etalchar{+}}25]{curr-omega}
Josh Alman, Ran Duan, Virginia {Vassilevska Williams}, Yinzhan Xu, Zixuan Xu, and Renfei Zhou.
\newblock {More Asymmetry Yields Faster Matrix Multiplication}.
\newblock In Yossi Azar and Debmalya Panigrahi, editors, {\em Proceedings of the 2025 Annual {ACM-SIAM} Symposium on Discrete Algorithms, {SODA} 2025, New Orleans, LA, USA, January 12-15, 2025}, pages 2005--2039. {SIAM}, 2025.
\newblock \href {https://doi.org/10.1137/1.9781611978322.63} {\path{doi:10.1137/1.9781611978322.63}}.

\bibitem[AFK{\etalchar{+}}24]{fast-bmm}
Amir Abboud, Nick Fischer, Zander Kelley, Shachar Lovett, and Raghu Meka.
\newblock {New Graph Decompositions and Combinatorial Boolean Matrix Multiplication Algorithms}.
\newblock In Bojan Mohar, Igor Shinkar, and Ryan O'Donnell, editors, {\em Proceedings of the 56th Annual {ACM} Symposium on Theory of Computing, {STOC} 2024, Vancouver, BC, Canada, June 24-28, 2024}, pages 935--943. {ACM}, 2024.
\newblock \href {https://doi.org/10.1145/3618260.3649696} {\path{doi:10.1145/3618260.3649696}}.

\bibitem[ALW14]{losing-ws}
Amir Abboud, Kevin Lewi, and Ryan Williams.
\newblock Losing weight by gaining edges.
\newblock In Andreas~S. Schulz and Dorothea Wagner, editors, {\em Algorithms - {ESA} 2014 - 22th Annual European Symposium, Wroclaw, Poland, September 8-10, 2014. Proceedings}, volume 8737 of {\em Lecture Notes in Computer Science}, pages 1--12. Springer, 2014.
\newblock \href {https://doi.org/10.1007/978-3-662-44777-2_1} {\path{doi:10.1007/978-3-662-44777-2_1}}.

\bibitem[AVW14]{clique-hyp-2}
Amir Abboud, Virginia {Vassilevska Williams}, and Oren Weimann.
\newblock {Consequences of Faster Alignment of Sequences}.
\newblock In Javier Esparza, Pierre Fraigniaud, Thore Husfeldt, and Elias Koutsoupias, editors, {\em Automata, Languages, and Programming - 41st International Colloquium, {ICALP} 2014, Copenhagen, Denmark, July 8-11, 2014, Proceedings, Part {I}}, volume 8572 of {\em Lecture Notes in Computer Science}, pages 39--51. Springer, 2014.
\newblock \href {https://doi.org/10.1007/978-3-662-43948-7_4} {\path{doi:10.1007/978-3-662-43948-7_4}}.

\bibitem[Bel62]{bellman_dp}
Richard Bellman.
\newblock Dynamic programming treatment of the travelling salesman problem.
\newblock {\em Journal of the ACM (JACM)}, 9(1):61--63, 1962.
\newblock \href {https://doi.org/10.1145/321105.321111} {\path{doi:10.1145/321105.321111}}.

\bibitem[BHK09]{set_partitioning}
Andreas Bj{\"{o}}rklund, Thore Husfeldt, and Mikko Koivisto.
\newblock Set partitioning via inclusion-exclusion.
\newblock {\em {SIAM} J. Comput.}, 39(2):546--563, 2009.
\newblock \href {https://doi.org/10.1137/070683933} {\path{doi:10.1137/070683933}}.

\bibitem[BHKK07]{fsc}
Andreas Bj{\"{o}}rklund, Thore Husfeldt, Petteri Kaski, and Mikko Koivisto.
\newblock {Fourier Meets M{\"{o}}bius: Fast Subset Convolution}.
\newblock In David~S. Johnson and Uriel Feige, editors, {\em Proceedings of the 39th Annual {ACM} Symposium on Theory of Computing, San Diego, California, USA, June 11-13, 2007}, pages 67--74. {ACM}, 2007.
\newblock \href {https://doi.org/10.1145/1250790.1250801} {\path{doi:10.1145/1250790.1250801}}.

\bibitem[Bib13]{ac_cs_view1}
Khodakhast Bibak.
\newblock {Additive Combinatorics: With a View Towards Computer Science and Cryptography - An Exposition}.
\newblock In Jonathan~M. Borwein, Igor~E. Shparlinski, and Wadim Zudilin, editors, {\em Number Theory and Related Fields, In Memory of Alf van der Poorten}, number Theory, pages 99--128. Springer, 2013.
\newblock \href {https://doi.org/10.1007/978-1-4614-6642-0_4} {\path{doi:10.1007/978-1-4614-6642-0_4}}.

\bibitem[Bj{\"{o}}10]{hamiltonicity}
Andreas Bj{\"{o}}rklund.
\newblock {Determinant Sums for Undirected Hamiltonicity}.
\newblock In {\em 51th Annual {IEEE} Symposium on Foundations of Computer Science, {FOCS} 2010, October 23-26, 2010, Las Vegas, Nevada, {USA}}, pages 173--182. {IEEE} Computer Society, 2010.
\newblock \href {https://doi.org/10.1109/FOCS.2010.24} {\path{doi:10.1109/FOCS.2010.24}}.

\bibitem[Bj{\"{o}}16]{below_all_subsets1}
Andreas Bj{\"{o}}rklund.
\newblock {Below All Subsets for Some Permutational Counting Problems}.
\newblock In Rasmus Pagh, editor, {\em 15th Scandinavian Symposium and Workshops on Algorithm Theory, {SWAT} 2016, June 22-24, 2016, Reykjavik, Iceland}, volume~53 of {\em LIPIcs}, pages 17:1--17:11. Schloss Dagstuhl - Leibniz-Zentrum f{\"{u}}r Informatik, 2016.
\newblock \href {https://doi.org/10.4230/LIPIcs.SWAT.2016.17} {\path{doi:10.4230/LIPIcs.SWAT.2016.17}}.

\bibitem[BT17]{clique-hyp-1}
Arturs Backurs and Christos Tzamos.
\newblock {Improving Viterbi is Hard: Better Runtimes Imply Faster Clique Algorithms}.
\newblock In Doina Precup and Yee~Whye Teh, editors, {\em Proceedings of the 34th International Conference on Machine Learning, {ICML} 2017, Sydney, NSW, Australia, 6-11 August 2017}, volume~70 of {\em Proceedings of Machine Learning Research}, pages 311--321. {PMLR}, 2017.
\newblock URL: \url{http://proceedings.mlr.press/v70/backurs17a.html}.

\bibitem[CFK{\etalchar{+}}15]{Cygan_book}
Marek Cygan, Fedor~V. Fomin, Lukasz Kowalik, Daniel Lokshtanov, D{\'{a}}niel Marx, Marcin Pilipczuk, Michal Pilipczuk, and Saket Saurabh.
\newblock {\em Parameterized Algorithms}.
\newblock Springer, 1st edition, 2015.
\newblock \href {https://doi.org/10.1007/978-3-319-21275-3} {\path{doi:10.1007/978-3-319-21275-3}}.

\bibitem[Cha02]{freiman-better-1}
Mei-Chu Chang.
\newblock {A polynomial bound in Freiman's theorem}.
\newblock {\em Duke Mathematical Journal}, 113(3):399 -- 419, 2002.
\newblock \href {https://doi.org/10.1215/S0012-7094-02-11331-3} {\path{doi:10.1215/S0012-7094-02-11331-3}}.

\bibitem[Cha13]{chan_shaving_logs}
Timothy~M. Chan.
\newblock {The Art of Shaving Logs}.
\newblock In Frank Dehne, Roberto Solis{-}Oba, and J{\"{o}}rg{-}R{\"{u}}diger Sack, editors, {\em Algorithms and Data Structures - 13th International Symposium, {WADS} 2013, London, ON, Canada, August 12-14, 2013. Proceedings}, volume 8037 of {\em Lecture Notes in Computer Science}, page 231. Springer, 2013.
\newblock \href {https://doi.org/10.1007/978-3-642-40104-6_20} {\path{doi:10.1007/978-3-642-40104-6_20}}.

\bibitem[CL15]{chan-lewenstein}
Timothy~M. Chan and Moshe Lewenstein.
\newblock {Clustered Integer 3SUM via Additive Combinatorics}.
\newblock In Rocco~A. Servedio and Ronitt Rubinfeld, editors, {\em Proceedings of the Forty-Seventh Annual {ACM} on Symposium on Theory of Computing, {STOC} 2015, Portland, OR, USA, June 14-17, 2015}, pages 31--40. {ACM}, 2015.
\newblock \href {https://doi.org/10.1145/2746539.2746568} {\path{doi:10.1145/2746539.2746568}}.

\bibitem[CMWW17]{min-plus-1}
Marek Cygan, Marcin Mucha, Karol Wegrzycki, and Michal Wlodarczyk.
\newblock {On Problems Equivalent to $(\min, +)$-Convolution}.
\newblock In Ioannis Chatzigiannakis, Piotr Indyk, Fabian Kuhn, and Anca Muscholl, editors, {\em 44th International Colloquium on Automata, Languages, and Programming, {ICALP} 2017, July 10-14, 2017, Warsaw, Poland}, volume~80 of {\em LIPIcs}, pages 22:1--22:15. Schloss Dagstuhl - Leibniz-Zentrum f{\"{u}}r Informatik, 2017.
\newblock \href {https://doi.org/10.4230/LIPIcs.ICALP.2017.22} {\path{doi:10.4230/LIPIcs.ICALP.2017.22}}.

\bibitem[CPPW14]{below_all_subsets2}
Marek Cygan, Marcin Pilipczuk, Michal Pilipczuk, and Jakub~Onufry Wojtaszczyk.
\newblock Scheduling partially ordered jobs faster than 2 n.
\newblock {\em Algorithmica}, 68(3):692--714, 2014.
\newblock \href {https://doi.org/10.1007/s00453-012-9694-7} {\path{doi:10.1007/s00453-012-9694-7}}.

\bibitem[CST01]{apx-ws}
Pierluigi Crescenzi, Riccardo Silvestri, and Luca Trevisan.
\newblock {On Weighted vs Unweighted Versions of Combinatorial Optimization Problems}.
\newblock {\em Inf. Comput.}, 167(1):10--26, 2001.
\newblock \href {https://doi.org/10.1006/inco.2000.3011} {\path{doi:10.1006/inco.2000.3011}}.

\bibitem[CVX23]{fredman_meets_domi}
Timothy~M. Chan, Virginia {Vassilevska Williams}, and Yinzhan Xu.
\newblock {Fredman's Trick Meets Dominance Product: Fine-Grained Complexity of Unweighted APSP, 3SUM Counting, and More}.
\newblock In Barna Saha and Rocco~A. Servedio, editors, {\em Proceedings of the 55th Annual {ACM} Symposium on Theory of Computing, {STOC} 2023, Orlando, FL, USA, June 20-23, 2023}, pages 419--432. {ACM}, 2023.
\newblock \href {https://doi.org/10.1145/3564246.3585237} {\path{doi:10.1145/3564246.3585237}}.

\bibitem[DW71]{dreyfus1971steiner}
Stuart~E. Dreyfus and Robert~A. Wagner.
\newblock The steiner problem in graphs.
\newblock {\em Networks}, 1(3):195--207, 1971.
\newblock \href {https://doi.org/10.1002/net.3230010302} {\path{doi:10.1002/net.3230010302}}.

\bibitem[EG04]{faster-uw-k-clique}
Friedrich Eisenbrand and Fabrizio Grandoni.
\newblock {On the complexity of fixed parameter clique and dominating set}.
\newblock {\em Theor. Comput. Sci.}, 326(1-3):57--67, 2004.
\newblock \href {https://doi.org/10.1016/j.tcs.2004.05.009} {\path{doi:10.1016/j.tcs.2004.05.009}}.

\bibitem[FGK09]{below_all_subsets3}
Fedor~V. Fomin, Fabrizio Grandoni, and Dieter Kratsch.
\newblock {A measure {\&} conquer approach for the analysis of exact algorithms}.
\newblock {\em J. {ACM}}, 56(5):25:1--25:32, 2009.
\newblock \href {https://doi.org/10.1145/1552285.1552286} {\path{doi:10.1145/1552285.1552286}}.

\bibitem[FKM{\etalchar{+}}07]{fuchs2007dynamic}
Bernhard Fuchs, Walter Kern, Daniel M{\"{o}}lle, Stefan Richter, Peter Rossmanith, and Xinhui Wang.
\newblock Dynamic programming for minimum steiner trees.
\newblock {\em Theory Comput. Syst.}, 41(3):493--500, 2007.
\newblock \href {https://doi.org/10.1007/s00224-007-1324-4} {\path{doi:10.1007/s00224-007-1324-4}}.

\bibitem[Fre73]{freiman1}
Gregory~A Freiman.
\newblock Foundations of a structural theory of set addition.
\newblock {\em Translation of Math. Monographs}, 37, 1973.

\bibitem[HK61]{held_karp_dp}
Michael Held and Richard~M. Karp.
\newblock A dynamic programming approach to sequencing problems.
\newblock In Thomas~C. Rowan, editor, {\em Proceedings of the 16th {ACM} national meeting, {ACM} 1961, {USA}}, volume~10, page~71. {ACM}, 1961.
\newblock \href {https://doi.org/10.1145/800029.808532} {\path{doi:10.1145/800029.808532}}.

\bibitem[JX23]{3sum_sidon}
Ce~Jin and Yinzhan Xu.
\newblock Removing additive structure in 3sum-based reductions.
\newblock In Barna Saha and Rocco~A. Servedio, editors, {\em Proceedings of the 55th Annual {ACM} Symposium on Theory of Computing, {STOC} 2023, Orlando, FL, USA, June 20-23, 2023}, pages 405--418. {ACM}, 2023.
\newblock \href {https://doi.org/10.1145/3564246.3585157} {\path{doi:10.1145/3564246.3585157}}.

\bibitem[Koi06]{koivisto-correct}
Mikko Koivisto.
\newblock Optimal 2-constraint satisfaction via sum-product algorithms.
\newblock {\em Inf. Process. Lett.}, 98(1):24--28, 2006.
\newblock \href {https://doi.org/10.1016/j.ipl.2005.11.013} {\path{doi:10.1016/j.ipl.2005.11.013}}.

\bibitem[KPS17]{min-plus-2}
Marvin K{\"{u}}nnemann, Ramamohan Paturi, and Stefan Schneider.
\newblock {On the Fine-Grained Complexity of One-Dimensional Dynamic Programming}.
\newblock In Ioannis Chatzigiannakis, Piotr Indyk, Fabian Kuhn, and Anca Muscholl, editors, {\em 44th International Colloquium on Automata, Languages, and Programming, {ICALP} 2017, July 10-14, 2017, Warsaw, Poland}, volume~80 of {\em LIPIcs}, pages 21:1--21:15. Schloss Dagstuhl - Leibniz-Zentrum f{\"{u}}r Informatik, 2017.
\newblock \href {https://doi.org/10.4230/LIPIcs.ICALP.2017.21} {\path{doi:10.4230/LIPIcs.ICALP.2017.21}}.

\bibitem[LN10]{saving_space}
Daniel Lokshtanov and Jesper Nederlof.
\newblock {Saving Space by Algebraization}.
\newblock In {\em Proceedings of the Forty-second ACM Symposium on Theory of Computing}, pages 321--330, 2010.
\newblock \href {https://doi.org/10.1145/1806689.1806735} {\path{doi:10.1145/1806689.1806735}}.

\bibitem[Lov17]{ac_cs_view2}
Shachar Lovett.
\newblock {Additive Combinatorics and its Applications in Theoretical Computer Science}.
\newblock {\em Theory Comput.}, 8:1--55, 2017.
\newblock \href {https://doi.org/10.4086/toc.gs.2017.008} {\path{doi:10.4086/toc.gs.2017.008}}.

\bibitem[LSS23]{min_k_cut}
Daniel Lokshtanov, Saket Saurabh, and Vaishali Surianarayanan.
\newblock {Breaking the All Subsets Barrier for Min $k$-Cut}.
\newblock In Kousha Etessami, Uriel Feige, and Gabriele Puppis, editors, {\em 50th International Colloquium on Automata, Languages, and Programming, {ICALP} 2023, July 10-14, 2023, Paderborn, Germany}, volume 261 of {\em LIPIcs}, pages 90:1--90:19. Schloss Dagstuhl - Leibniz-Zentrum f{\"{u}}r Informatik, 2023.
\newblock \href {https://doi.org/10.4230/LIPIcs.ICALP.2023.90} {\path{doi:10.4230/LIPIcs.ICALP.2023.90}}.

\bibitem[LVW18]{clique-hyp-3}
Andrea Lincoln, Virginia {Vassilevska Williams}, and R.~Ryan Williams.
\newblock {Tight Hardness for Shortest Cycles and Paths in Sparse Graphs}.
\newblock In Artur Czumaj, editor, {\em Proceedings of the Twenty-Ninth Annual {ACM-SIAM} Symposium on Discrete Algorithms, {SODA} 2018, New Orleans, LA, USA, January 7-10, 2018}, pages 1236--1252. {SIAM}, 2018.
\newblock \href {https://doi.org/10.1137/1.9781611975031.80} {\path{doi:10.1137/1.9781611975031.80}}.

\bibitem[MRR06]{faster_steiner1}
Daniel M{\"{o}}lle, Stefan Richter, and Peter Rossmanith.
\newblock {A Faster Algorithm for the Steiner Tree Problem}.
\newblock In Bruno Durand and Wolfgang Thomas, editors, {\em {STACS} 2006, 23rd Annual Symposium on Theoretical Aspects of Computer Science, Marseille, France, February 23-25, 2006, Proceedings}, volume 3884 of {\em Lecture Notes in Computer Science}, pages 561--570. Springer, 2006.
\newblock \href {https://doi.org/10.1007/11672142_46} {\path{doi:10.1007/11672142_46}}.

\bibitem[Ned20]{nederlof_tsp}
Jesper Nederlof.
\newblock Bipartite {TSP} in $o(1.9999^n)$ time, assuming quadratic time matrix multiplication.
\newblock In Konstantin Makarychev, Yury Makarychev, Madhur Tulsiani, Gautam Kamath, and Julia Chuzhoy, editors, {\em Proceedings of the 52nd Annual {ACM} {SIGACT} Symposium on Theory of Computing, {STOC} 2020, Chicago, IL, USA, June 22-26, 2020}, pages 40--53. {ACM}, 2020.
\newblock \href {https://doi.org/10.1145/3357713.3384264} {\path{doi:10.1145/3357713.3384264}}.

\bibitem[NP85]{uw-k-clique}
Jaroslav Nešetřil and Svatopluk Poljak.
\newblock {On the complexity of the subgraph problem}.
\newblock {\em Commentationes Mathematicae Universitatis Carolinae}, 026(2):415--419, 1985.
\newblock URL: \url{http://eudml.org/doc/17394}.

\bibitem[RW24]{esa-paper}
Tim Randolph and Karol Wegrzycki.
\newblock {Parameterized Algorithms on Integer Sets with Small Doubling: Integer Programming, Subset Sum and $k$-SUM}.
\newblock In Timothy~M. Chan, Johannes Fischer, John Iacono, and Grzegorz Herman, editors, {\em 32nd Annual European Symposium on Algorithms, {ESA} 2024, September 2-4, 2024, Royal Holloway, London, United Kingdom}, volume 308 of {\em LIPIcs}, pages 96:1--96:19. Schloss Dagstuhl - Leibniz-Zentrum f{\"{u}}r Informatik, 2024.
\newblock \href {https://doi.org/10.4230/LIPIcs.ESA.2024.96} {\path{doi:10.4230/LIPIcs.ESA.2024.96}}.

\bibitem[San12a]{freiman-better-3}
Tom Sanders.
\newblock {On the {Bogolyubov}-{Ruzsa} lemma}.
\newblock {\em Anal. PDE}, 5(3):627--655, 2012.
\newblock \href {https://doi.org/10.2140/apde.2012.5.627} {\path{doi:10.2140/apde.2012.5.627}}.

\bibitem[San12b]{freiman-better-4}
Tom Sanders.
\newblock The structure theory of set addition revisited.
\newblock {\em Bulletin of the American Mathematical Society}, 50(1):93–127, October 2012.
\newblock URL: \url{http://dx.doi.org/10.1090/S0273-0979-2012-01392-7}, \href {https://doi.org/10.1090/s0273-0979-2012-01392-7} {\path{doi:10.1090/s0273-0979-2012-01392-7}}.

\bibitem[Sch11]{freiman-better-2}
Tomasz Schoen.
\newblock {Near optimal bounds in Freiman's theorem}.
\newblock {\em Duke Mathematical Journal}, 158(1):1 -- 12, 2011.
\newblock \href {https://doi.org/10.1215/00127094-1276283} {\path{doi:10.1215/00127094-1276283}}.

\bibitem[Sto24]{approx-min-sum-conv}
Mihail Stoian.
\newblock Approximate min-sum subset convolution.
\newblock In Marcin Bienkowski and Matthias Englert, editors, {\em Approximation and Online Algorithms - 22nd International Workshop, {WAOA} 2024, Egham, UK, September 5-6, 2024, Proceedings}, volume 15269 of {\em Lecture Notes in Computer Science}, pages 198--212. Springer, 2024.
\newblock \href {https://doi.org/10.1007/978-3-031-81396-2_14} {\path{doi:10.1007/978-3-031-81396-2_14}}.

\bibitem[Tre09]{trevisan_ac}
Luca Trevisan.
\newblock Guest column: additive combinatorics and theoretical computer science.
\newblock {\em {SIGACT} News}, 40(2):50--66, 2009.
\newblock \href {https://doi.org/10.1145/1556154.1556170} {\path{doi:10.1145/1556154.1556170}}.

\bibitem[TV06]{tao_vu}
Terence Tao and Van~H Vu.
\newblock {\em Additive combinatorics}, volume 105.
\newblock Cambridge University Press, 2006.

\bibitem[Vio11]{ac_cs_view3}
Emanuele Viola.
\newblock Selected results in additive combinatorics: An exposition.
\newblock {\em Theory Comput.}, 3:1--15, 2011.
\newblock \href {https://doi.org/10.4086/toc.gs.2011.003} {\path{doi:10.4086/toc.gs.2011.003}}.

\bibitem[vRNvD09]{below_all_subsets4}
Johan M.~M. van Rooij, Jesper Nederlof, and Thomas~C. van Dijk.
\newblock Inclusion/exclusion meets measure and conquer.
\newblock In Amos Fiat and Peter Sanders, editors, {\em Algorithms - {ESA} 2009, 17th Annual European Symposium, Copenhagen, Denmark, September 7-9, 2009. Proceedings}, volume 5757 of {\em Lecture Notes in Computer Science}, pages 554--565. Springer, Springer, 2009.
\newblock \href {https://doi.org/10.1007/978-3-642-04128-0_50} {\path{doi:10.1007/978-3-642-04128-0_50}}.

\bibitem[VX20]{special-min-plus-prod}
Virginia {Vassilevska Williams} and Yinzhan Xu.
\newblock {Truly Subcubic Min-Plus Product for Less Structured Matrices, with Applications}.
\newblock In Shuchi Chawla, editor, {\em Proceedings of the 2020 {ACM-SIAM} Symposium on Discrete Algorithms, {SODA} 2020, Salt Lake City, UT, USA, January 5-8, 2020}, pages 12--29. {SIAM}, 2020.
\newblock \href {https://doi.org/10.1137/1.9781611975994.2} {\path{doi:10.1137/1.9781611975994.2}}.

\bibitem[W{\k{e}}g19]{karol_phd}
K.~W{\k{e}}grzycki.
\newblock {\em Provably Optimal Dynamic Programming}.
\newblock 2019.
\newblock URL: \url{https://books.google.de/books?id=UmYfzwEACAAJ}.

\bibitem[Wil05]{maxcut-williams}
Ryan Williams.
\newblock A new algorithm for optimal 2-constraint satisfaction and its implications.
\newblock {\em Theor. Comput. Sci.}, 348(2-3):357--365, 2005.
\newblock \href {https://doi.org/10.1016/j.tcs.2005.09.023} {\path{doi:10.1016/j.tcs.2005.09.023}}.

\bibitem[Zwi02]{zwick_apsp}
Uri Zwick.
\newblock All pairs shortest paths using bridging sets and rectangular matrix multiplication.
\newblock {\em J. {ACM}}, 49(3):289--317, may 2002.
\newblock \href {https://doi.org/10.1145/567112.567114} {\path{doi:10.1145/567112.567114}}.

\end{thebibliography}
\bibliographystyle{alphaurl}

\appendix

\newpage

\section{Pairing Function's Monomorphism}\label{appendix:mono}

\PAIRING*
\begin{proof}
    We first prove that $\kappa$ is a homomorphism and then prove that it is also injective.~\\\newline
    \textbf{Homomorphism.}~Let $\alpha = \tup{\alpha_1, \ldots, \alpha_d}$ and $\beta = \tup{\beta_1, \ldots, \beta_d}$ be two GAP-coordinates. Their entrywise addition $\alpha \oplus \beta$ is defined as\[
        \alpha \oplus \beta = \tup{\alpha_1 + \beta_1, \ldots, \alpha_d + \beta_d}.
    \]
    At this point, we use the fact that $D$ is well-defined, i.e., the dimension bounds are not exceeded: $\alpha_i + \beta_i \leq L_i, \forall i \in [d]$. (This is guaranteed by the fact that we enlarge the dimension bounds before running the bounded-input algorithm; compare line~\ref{line:3} in Alg.~\ref{algo:meta}). Formally, we have to prove that
    \begin{equation}
        \label{eq:kappa_to_prove}
        \kappa\left(d, \tup{L_i}_{i\in [d]}, \alpha \oplus \beta\right) = \kappa\left(d, \tup{L_i}_{i\in [d]}, \alpha\right) + \kappa\left(d, \tup{L_i}_{i\in [d]}, \beta\right).
    \end{equation}
    In the sequel, we use the notation $\kappa_{d-1}(x)$ to denote that $\kappa$ is applied only on the first $d - 1$ dimensions of a tuple of GAP coordinates $x$, i.e., $\kappa\left(d-1, \tup{L_i}_{i\in[d-1]}, \tup{x_i}_{i\in[d-1]}\right)$; analogously, $\kappa_{d}$ will refer to the first $d$ dimensions.
    To prove Eq.~\eqref{eq:kappa_to_prove}, we use induction on the number of dimensions. The base case, $d = 1$, follows by construction. Otherwise, assume Eq.~\eqref{eq:kappa_to_prove} holds for $d - 1$ dimensions (IH1), i.e.,
    \begin{align*}
        &\kappa_{d-1}(\alpha \oplus \beta) = \kappa_{d-1}(\alpha) + \kappa_{d-1}(\beta).
    \end{align*}
    Then,
    \begin{flalign*}
        &&\kappa_{d}\left(\alpha \oplus \beta\right) &= \kappa_{d}\left(\alpha\right) + \kappa_{d}\left(\beta\right)\\
        \overset{\kappa}{\iff}&& (\alpha \oplus \beta)_d + (L_d+1)\kappa_{d-1}\left(\alpha \oplus \beta\right) &= \alpha_d + \beta_d + (L_d+1)\left(\kappa_{d-1}(\alpha) + \kappa_{d-1}(\beta)\right)\\
        \overset{\text{IH1}}{\iff}&& (\alpha \oplus \beta)_d &= \alpha_d + \beta_d,
    \end{flalign*}
    which is exactly the definition of $\alpha \oplus \beta$.~\\\newline
    \textbf{Injectivity.}~We next prove the injectivity of $\kappa$, by showing that $\kappa$ actually satisfies a certain monotonicity property, outlined in the following:
    \begin{equation}
        \alpha \prec \beta \implies \kappa_d(\alpha) < \kappa_d(\beta),
        \label{eq:monotonicity}
    \end{equation}
    where ``$\alpha\prec\beta$'' means that $\alpha$ is strictly lexicographically smaller than $\beta$.
    We prove this property by induction on the number of dimensions. The base case, $d = 1$ follows by construction.
    Now, assume that property Eq.~\eqref{eq:monotonicity} holds for the first $d - 1$ dimensions (IH2), i.e., \[
        \tup{\alpha_i}_{i\in[d-1]} \prec \tup{\beta_i}_{i\in[d-1]} \implies \kappa_{d-1}(\alpha) < \kappa_{d-1}(\beta).
    \]
    To prove the original Eq.~\eqref{eq:monotonicity}, we differentiate between two cases when $\alpha \prec \beta$ holds true:
    \begin{description}
        \item[Case] $\tup{\alpha_i}_{i\in[d-1]} = \tup{\beta_i}_{i\in[d-1]}$. This means that $\kappa_{d-1}(\alpha) = \kappa_{d-1}(\beta)$. Moreover, $\alpha_d < \beta_d$ ($\ast$), otherwise $\alpha \prec \beta$ cannot be true. Hence,
        \begin{align*}
             &&\kappa_d(\alpha) &< \kappa_d(\beta)\\
        \overset{\kappa}{\iff}&& \alpha_d + (L_d+1)\kappa_{d-1}(\alpha) &< \beta_d + (L_d+1)\kappa_{d-1}(\beta)\\
        \oset{\text{Case}}{\iff}&& \alpha_d &< \beta_d,
        \end{align*}
        which is true due to ($\ast$).
        \item[Case] $\tup{\alpha_i}_{i\in[d-1]} \prec \tup{\beta_i}_{i\in[d-1]}$. In this case, there is no relation between $\alpha_d$ and $\beta_d$.\footnote{To see why this is the case, note that in an English dictionary ``root'' $\prec$ ``rope'', even though the last letter of ``rope'', ``e'', comes before the ``t'' of ``root''.} Hence, we can apply IH2 and use the fact that the difference $\kappa_{d-1}(\beta) - \kappa(\alpha)$ is strictly greater than 0:
        \begin{align*}
             &&\kappa_d(\alpha) &< \kappa_d(\beta)\\
        \overset{\kappa}{\iff}&& \alpha_d + (L_d+1)\kappa_{d-1}(\alpha) &< \beta_d + (L_d+1)\kappa_{d-1}(\beta)\\
        \iff&& \alpha_d - \beta_d &< (L_d+1)\left(\kappa_{d-1}(\beta) - \kappa_{d-1}(\alpha)\right)\\
        \overset{\text{IH2}}{\iff}&& \alpha_d - \beta_d &< L_d + 1,
        \end{align*}
        which is true since the difference $\alpha_d - \beta_d$ can only lie in $[-L_d, L_d]$.        
    \end{description}
    This completes the injectivity argument, hence $\kappa$ is a monomorphism. 
\end{proof}

\section{Edge-Weighted $k$-Clique with Small Weights}\label{appendix:NP}.

We show that Nešetřil-Poljak's algorithm~\cite{uw-k-clique} can be extended to support small weights. We consider the $3k$-case here; both the $3k + 1$- and $3k + 2$-cases are handled analogously. 

\begin{restatable}[]{lemma}{NP}
    \label{lemma:np}
    The \textsc{Edge-Weighted $3k$-Clique} problem on undirected graphs $G = (V, E)$ with integer edge weights in $[0, W]$ can be solved in time $O^*(n^{\omega k}\,W)$.
\end{restatable}
\begin{proof}
We extend Nešetřil--Poljak's construction to the edge-weighted setting. Let $H$ be the auxiliary graph whose vertices are all $k$-cliques $Y\subseteq V$ of $G$, and where
distinct $Y_1,Y_2$ are adjacent iff $Y_1\cup Y_2$ induces a $2k$-clique in $G$. Hence, $H$ has $N=\Theta(n^k)$ vertices.

Denote by $w(e)\in[0,W]$ the weight of an edge $e \in E(G)$. For a $k$-clique $Y$, set $w'(Y) = w(E(G[Y]))$, the ``internal'' weight of node $Y$, and for adjacent $Y_1,Y_2$ in $H$, set
\begin{equation*}
    w'(Y_1,Y_2) = \sum_{\substack{u\in Y_1\\ v\in Y_2}} w(uv),
\end{equation*}
representing the ``cross'' weight between $Y_1$ and $Y_2$.
If $A,B,C$ form a triangle in $H$, the total weight of the corresponding $3k$-clique
$A\cup B\cup C$ is
\begin{equation*}
    W_{\mathrm{clique}} = w'(A)+w'(B)+w'(C)+w'(A,B)+w'(B,C)+w'(C,A).
\end{equation*}
Now, since we would like to have edge weights only (in order to not have a specialized algorithm for coping with both cases), redefine the weight of the edge $\{Y_1, Y_2\}$ as:
\begin{equation*}
    w'(Y_1, Y_2) = 2 \sum_{\substack{u\in Y_1\\ v\in Y_2}} w(uv) + w'(Y_1) + w'(Y_2),
\end{equation*}
that is, we double the previously edge weight, and add the internal weights of the two nodes. Hence, for every triangle $\{A, B, C\}$ in $H$, we have
\begin{equation*}
    w'(A, B) + w'(B, C) + w'(C, A) = 2 W_{\mathrm{clique}},
\end{equation*}
so minimizing triangle weight in $H$ is equivalent to minimizing $3k$-clique weight in $G$.

Next, note that the minimizing edge weight on $H$ satisfies
\begin{equation*}
    W_H^{\max} \leq \big(2k^2 + 2\textstyle{\binom{k}{2}}\big)\,W = (3k^2-k) W = \Theta(k^2 W).
\end{equation*}
A minimum-weight triangle in an $N$-vertex edge-weighted graph with integer weights in $[0, W_H^{\max}]$ can be found via a single distance $(\min, +)$-product in time $\tilde O(W_H^{\max}\,N^{\omega})$. Instantiating $N=\Theta(n^k)$ and $W_H^{\max}=\Theta(k^2W)$ yields time $O^*(n^{\omega k}\,W)$. Finally, explicitly constructing $H$ and all $w'(\cdot,\cdot)$ values takes $O(n^{2k}\operatorname{poly}(k) \log W)$-time, which is dominated by $n^{\omega k}$ since $\omega > 2$.
\end{proof}

\end{document}